\documentclass{article}

\usepackage{arxiv}

\usepackage[utf8]{inputenc} 
\usepackage[T1]{fontenc}    
\usepackage{hyperref}       
\usepackage{url}            
\usepackage{booktabs}       
\usepackage{amsfonts}       
\usepackage{nicefrac}       
\usepackage{microtype}      
\usepackage{lipsum}

\usepackage{amsthm}

\usepackage{amsmath}
\usepackage{graphicx,psfrag,epsf}
\usepackage{enumerate}
\usepackage{natbib}
\usepackage{url} 
\usepackage{algpseudocode}
\usepackage{comment}
\usepackage{amssymb,amsmath}
\usepackage{etex}
\usepackage[all]{xypic}
\usepackage{bm}
\usepackage{tikz}
\usetikzlibrary{arrows}
\usepackage{bmpsize}
\usepackage[utf8]{inputenc} 
\newtheorem{prop}{Proposition}
\def\gr{}
\def\bbeta{\boldsymbol{\beta}}
\def\bl{\boldsymbol{\lambda}}

\def\kk{\kappa_k}
\def\bk{\boldsymbol{\kappa}}

\def\Z{{\mathbb Z}}

\def\phi{\emptyset}

\def\bX{\boldsymbol{X}}
\def\bC{\boldsymbol{C}}
\def\bc{\boldsymbol{c}}

\title{Synthetic likelihood method for reaction network inference}

\author{
  Daniel F.~Linder\\
 Division of Biostatistics and Data Science\\
  Medical College of Georgia\\
  Augusta University \\
  Augusta, GA 30912 \\
  \texttt{dlinder@augusta.edu} \\
   \And
 Grzegorz A.~Rempa{\l}a \\
Division of Biostatistics\\Mathematical Biosciences Institute\\
 Ohio State University \\
 Columbus, OH 43210 \\
   \texttt{rempala.3@osu.edu} \\
}

\begin{document}
\maketitle

\begin{abstract}
\gr{We propose a novel Markov chain Monte-Carlo (MCMC) method for reverse engineering the topological structure of stochastic reaction networks, a notoriously challenging problem that is relevant in many modern areas of research, like   discovering gene regulatory networks  or  analyzing epidemic  spread.  The method relies  on projecting the original time series trajectories onto information rich summary statistics and  constructing  the appropriate synthetic likelihood function to estimate reaction rates. The  resulting estimates are  consistent in the large volume limit  and are obtained without employing complicated tuning strategies and expensive resampling as  typically  used  by likelihood-free MCMC and approximate Bayesian methods. To illustrate run time improvements that can be achieved with our  approach, we present  a simulation study on inferring rates in a stochastic dynamical system arising from a density dependent Markov jump process. We then  apply the method to two real data examples:  the  RNA-seq data from  zebrafish experiment  and  the  incidence data from 1665 plague outbreak at Eyam, England.}

\end{abstract}


\section{Introduction}
\label{sec:intro}
 Recent developments in molecular technologies have allowed us to perform complex  biological experiments  aiming   at learning  the principles of signaling networks in living organisms.   For instance,  the knowledge of  a biochemical network of  physiological processes of living cells offers insights into developing targeted therapies for a wide range of  diseases, such as cancer, diabetes and more, by suggesting possible targets in gene pathways. In addition, decoding these mechanisms in organisms with regeneration capabilities, like the  zebrafish, can potentially help biologists at least partially answer questions about why humans lack these abilities. However, reverse engineering biochemical networks from observed data has  proven to be challenging from both statistical and computational standpoints \citep{oates2012network,daniels2015automated}. Indeed, the high-throughput molecular technology pushes the boundary of a classical statistical inferential paradigm, as the molecular quantification methods such as next generation sequencing,  cellular flow cytometry, and fluorescence microscopy (see, e.g. \cite{Perez:2004fk,Wheeler:2008ys}),  provide large amounts of high dimensional, longitudinal data from partially observed and poorly understood biochemical systems.

The goal of most statistical inference problems in this setting is either one of structure (topology)  or parameter estimation \citep{oates2014causal}.  Although related,  these two inference areas  are often considered  separately. The first addresses the structure of the underlying biochemical network, i.e. which chemical species directly alters the production rate of another.  In this paradigm the object of interest is usually a directed graph or a binary adjacency matrix describing the gene regulation structure of the system. Bayesian networks \citep{pearl1985bayesian, pearl1986fusion}, and, more recently, dynamic Bayesian networks, have been used for learning topology of such models since their graphical structure, which is encoded through distributional assumptions, is in essence the desired output \citep{morrissey2010reverse, morrissey2011inferring}. Bayesian networks and other similar graphical models are appealing because they provide biologists with a simple visual representation of the network. However, the fundamental disadvantage of simple graphical approaches is their reduction of a detailed kinetic view of the system to a set of often unrealistic assumptions on the relationships between graph nodes (representing either chemical or molecular species, like genes). For instance, although it is known that the interactions of species in biological systems typically exhibit a high degree of nonlinearity, it is often assumed that nodes are hierarchically (acyclically) linearly related, with the conditional Gaussian distributions. A practical disadvantage of these methods is that their solutions are of non-polynomial time complexity, and hence algorithms used for model fitting lead to suboptimal solutions, for instance by requiring greedy algorithms. Markov chain Monte Carlo (MCMC) routines have been shown to give moderate improvement in terms of finding optimal network structure, but computational issues still persist. 

 In addition to network structure or  topology estimation, the  second related inference area is that of  network parameters (e.g.,  reaction rates) estimation.  Typically, the methods here focus on fitting detailed kinetic models to the data. Their appeal lies in the fact that they rely on more precise representations of the underlying dynamical system \citep{fearnhead2014inference,golightly2005bayesian,golightly2012,oates2012network}. Their disadvantages are numerous however, particularly of the methods  based on exact likelihoods, where the  inference is usually not feasible for systems of relevant sizes, i.e. hundreds to thousands of species and reactions. For that reason the corresponding Bayesian frameworks based on the exact likelihoods are usually not applicable to topology estimation \citep{boys2008bayesian,choi2012inference}. Methods based on approximate likelihood, like the linear noise approximation (or LNA, see \citep{komorowski2009bayesian}), fare somewhat better, but they also usually do not allow for efficient posterior sampling due to the complicated form of the likelihood approximation  (see, for instance, 
 \citep{Linder:2015aa}).
 
In this paper, we develop a fully Bayesian inference  framework  for stochastic reaction networks \gr {on both structure  and parameters} given time course trajectories from the stochastic dynamical systems under mass action kinetics. The method uses summary statistics to form a synthetic likelihood following  the ideas presented by \cite{Wood:2010aa}. The advantage of the proposed methodology is that the synthetic likelihood is based on detailed kinetic models, and its form permits topology and parameter estimation simultaneously, instead of  separately. Our approach is Bayesian and allows efficient computation of the posterior network structure through point mass priors on parameters. An outline of the paper beyond the current section is as follows. In Section 2 we describe the types of dynamical systems under consideration. In Section 3 we describe the synthetic likelihood, the relevant priors, and outline an efficient MCMC algorithm to sample posteriors. The model fitting  is performed using  in silico  reaction network data in Section 4 and RNA-seq data from controlled experiments in the zebrafish in Section 5.  In Section~5 we also give an  example of application of our framework  beyond biochemical systems  by analyzing    historic  data  from   the  1665-1666  plague outbreak in  Eyam, England. The technical details of the MCMC algorithm derivations and R code implementation are available in the online Supplementary Material.

\section{Stochastic reaction network system}
The reaction systems we consider consist of $s$ chemical species, along with a set of  $r$ reaction channels where commonly $r > s$. We denote the system state at time $t$ as the vector $X(t)$ of dimension $s$, containing molecular counts of each species. The constant $\kappa_k \ge 0$ gives the reaction rate or speed of the $k^{th}$ reaction, $k=1, \dots, r$. When the $k^{th}$ reaction occurs at time $t$, the system transitions according to the integer valued vector  $\nu_k^\prime - \nu_k$, where $\nu^\prime_k$ is a vector of non-negative integers representing the number of species produced by reaction $k$ and $\nu_k$ representing the number consumed, i.e. 
\begin{equation*}X(t)=X(t-)+\nu_k^\prime-\nu_k. \end{equation*}
where $X(t-)$ is the system state at the instantaneous time before $t$. We denote by $n$   the system volume, typically the physical volume of the container (e.g., cell) times Avagadro's number, and by $Y_k$ the unit Poisson processes. Our further analysis is on systems that are well-stirred and thermally equilibrated, with processes obeying the classical mass action rate laws, \citep[see, e.g.,][]{gillespie1992rigorous} corresponding combinatorially to the number of different ways we can choose molecular substrates needed for the reaction $k$ to occur \citep[chapter 10]{ethier2009markov}. Defining $|\nu_k|=\sum_i \nu_{ik}$, the rates are
\begin{equation*}
\lambda^{(n)}_k(x)=n\kappa_k\frac{\prod_i\nu_{ik}!\binom{x_i}{\nu_{ik}}}{n^{|\nu_{k}|}}.
\end{equation*}
The nonhomogenous Poisson process with the above propensity function (see  also \citep{gillespie1992rigorous}) gives the system time-evolution equation
\begin{equation}\label{eq:2}
X(t)=X(0)+\sum_kY_k(\int_0^t\lambda_k^{(n)}(X(s))ds)(\nu_{k}^\prime-\nu_k).
\end{equation}
The model in Equation \eqref{eq:2} is often considered the most accurate representation of true system dynamics, and is in the general class of density dependent Markov jump processes (DDMJP). While the class of DDMJP models are often used to describe a wide variety of physical systems, like gene regulatory networks and stochastic epidemics, unfortunately the corresponding inference is complicated by highly intractable exact likelihoods. 

Letting $c=n^{-1}x$ we obtain species concentrations \gr{(say,   in moles per unit volume or relative density)}. The asymptotic notion of a large volume limit represents the system's  behavior as its volume increases to infinity while the species numbers are kept at  constant concentrations. This gives the classical deterministic law of mass action ordinary differential equation (ODE), which is referred in the chemical literature \citep{van1992stochastic} as the {\em reaction rate equation} 
\begin{equation}\label{eq:rre}
\dot{c}(t)=\sum_k\kappa_k \prod_i c^{\nu_{ik}}_i(t)(\nu_k^\prime-\nu_k).
\end{equation}  
The solution of the above ODE $c^{\bbeta}$ is parameterized by a vector $\bbeta$, a linear combination of the kinetic rates $\kappa_k$ of interest, as well as the initial condition $c(0)$. In what follows we will focus on estimation of $\bbeta$ under the assumption that it is a linear transformation of the $\kappa_k$'s. It is well known that identifiability of reaction networks is a nontrivial problem and is only guaranteed when certain reaction vectors are linearly independent for each source complex \citep{craciun2008identifiability}. However, the reparameterization from $\bk$ to $\bbeta$ can be done so as to ensure that $\bbeta$ is identifiable, as is the case for the examples we consider below. Results in \cite{Remp12} show that the least squares estimator, $\hat{\bbeta}$, which minimizes the distance between the data and the solution to \eqref{eq:rre} is consistent and asymptotically normal. These asymptotic properties are also true for solutions to the so-called {\em martingale estimating functions}, which are a generalization of the least squares estimator in this case \citep{bibby1995martingale}. Both methods produce statistics that are easily obtained from time course trajectories of the system. In what follows, we use the asymptotic properties of these statistics \gr{and their estimating equations}  to form a synthetic likelihood. The synthetic likelihood serves as a surrogate for, often intractable, exact likelihood and may be used in the same way to perform the usual  Bayesian inference.    

\section{Synthetic likelihood}\label{sec:lfm}  Ideally, parametric inference would be based on the likelihood function since, under the typical regularity conditions, maximum likelihood estimates enjoy good asymptotic properties, such as consistency and efficiency.  \gr {The likelihood approach also gives one  the ability to perform Bayesian inference. 
Unfortunately, the usage of  exact likelihood methods for parameter estimation in stochastic biochemical networks  faces some challenges due to the need for   computationally demanding routines, like, for instance, the particle filters \citep{Golightly:2011aa}. For that reason  many authors  have  focused  on approximate likelihood methods for reaction networks \citep{fearnhead2014inference,golightly2005bayesian,golightly2012,oates2012network}. However,  major practical limitation of many such methods, for instance,  approximate Bayesian computation (ABC) is their slow convergence and poor mixing in high dimensional problems  \citep{csillery2010approximate}.  To circumvent these technical complication  we propose here  an alternative method that is based on the idea  of  synthetic likelihood \citep{Wood:2010aa}.  }   

To introduce some notation, consider the data $D_j, \hspace{2mm} j=1,\ldots, N$ or the $j^{th}$ system trajectory consisting of $X_{ij}^{(n,\beta)} \in \Z^s_{\ge 0}$, which are the observed counts of  $s$ species, measured at a discrete grid of timepoints, $t_{ij}$, ($t_{1j},\dots,t_{m_jj}=T_j<\infty)$ not necessarily equidistant across trajectories, and with possibly different endpoints; i.e., $t_{m_jj}$ not necessarly equal to $t_{m_kk}$.  We assume that this observed data arrives from trajectories of the process for which the system volume $n$ is fixed and known and define  the concentration values as $C_n(t_{ij})=\bX_j(t_{ij})/n$. The LSE for the $j^{th}$ observation is then any solution of the optimization problem 
\begin{equation}\label{eq:LSE}
\hat{\bbeta}_j=\underset{\bbeta}{\textrm{argmin}}  \sum_{t_{ij}} ||C_n(t_{ij})-c^{\bbeta}(t_{ij})||_2^2
\end{equation}
or equivalently any solution to the following estimating equation
\begin{equation}\label{eq:EE}
 \sum_{t_{ij}} \partial c^{\bbeta}(t_{ij})(C_n(t_{ij})-c^{\bbeta}(t_{ij}))=0.
\end{equation}
Asymptotic properties, (as $n \rightarrow \infty $), and the regularity conditions for the consistency and normality of these solutions were discussed in \cite{Remp12}, with all systems under consideration here satisfying these conditions. The expression above is similar in form to the generalized estimating equations (GEEs) \citep{Zeger:1986aa}. GEEs have been used extensively for correlated and longitudinal data where a parametric form of the mean is known but the likelihood function is not readily available. The appeal of GEE estimates is that they exhibit many of the same properties of maximum likelihood estimates \citep{Zeger:1986aa}, even when the correlation structure of the dependent observations is misspecified. \\
\indent The ideas from GEEs were extended to discretely observed diffusion processes in \cite{bibby1995martingale} by considering so-called martingale estimating functions (MEF). Defining the filtration $\mathcal{F}_{ij}=\sigma(X(t_{1j}),\dots,X(t_{ij}))$, scaling the process by $n$, consider all zero-mean $P_{\bbeta}$-martingale estimating functions of the form 
\begin{equation}\label{eq:MEF}
G_j(\bbeta)=\sum_{t_{ij}} g_{(i-1)j}(\bbeta)(C_n(t_{ij})-F_i(C_n(t_{i-1j}),\bbeta))
\end{equation}
where $g_{(i-1)j}$ is $\mathcal{F}_{(i-1)j}$ measurable and $F_i(x,\bbeta)=E_{\bbeta}(C_n(t_{ij})|x)$. It was also shown in  \cite{bibby1995martingale} that the optimal estimating function, in the sense of the smallest asymptotic dispersion and where the data is assumed to arise from discrete observations of a diffusion, is of the form $g_{(i-1)j}=\dot{F}_i(C_n(t_{(i-1)j}),\bbeta)\rho_i(C_n(t_{(i-1)j}),\bbeta)^{-1}$ with $\rho_i(C_n(t_{(i-1)j}),\bbeta)=Var_{\beta}(C_n(t_{ij})|C_n(t_{(i-1)j}))$. The optimal estimating function may be approximated by, $\tilde{g}_{(i-1)j}$, with the substitutions $\tilde{F}_i(x,\bbeta)=c^{\bbeta}(t_{ij})$ as well as  $\tilde{\rho}_i(C_n(t_{(i-1)j}),\bbeta)=\hat{Var}_{\bbeta}(C_n(t_{ij})|C_n(t_{(i-1)j}))$. Fitting may  then be done iteratively by first updating the weights in \eqref{eq:MEF} for the current value of   $\bbeta=\bbeta^{(k-1)}$  with  $\tilde{g}_{(i-1)j}(\bbeta^{(k-1)})$ and  then updating $\bbeta=\bbeta^{(k)}$ by solving \eqref{eq:MEF} and repeating this process until convergence, or by replacing $\tilde{\rho}$ with its empirical estimate.

\subsection{Form of the Synthetic  Likelihood} 
Given the partially observed trajectory, $\boldsymbol{X}_j$, we denote by $\hat{\bbeta}_j$ the solution of either \eqref{eq:EE} or \eqref{eq:MEF}. In the mass action setting, the ODE coefficients $\bbeta$ are linear combinations of the reaction rates and can be written as $\bbeta=Q\bk$ where $Q$ is a $d\times r$ matrix (see, e.g., \citep{Linder:2013aa}). \gr{Some  specific examples of $\bbeta$ and  $Q$ are given in Section 4 below.}  It is straightforward to show that $\hat{\bbeta}_j$ is asymptotically Gaussian; i.e., $\sqrt{n}(\hat{\beta}-Q\bk) \rightarrow \mathcal{N}(0,\Sigma)$ as previously mentioned, see Supplementary Material.
The normality of $\hat{\beta}$  allows to  express  the synthetic likelihood for the $j^{th}$ trajectory (replicate) as  
\begin{equation}\label{eq:SL}
SL_j(\bk, \Sigma|\hat{\bbeta}_j):=f(\hat{\bbeta}_j|Q \bk, \Sigma)=(2\pi)^{-d/2}|\Sigma/n|^{-1/2}\exp\bigg\{-\frac{1}{2}(Q\bk-\hat{\bbeta}_j)^\top (\Sigma/n)^{-1}(Q\bk-\hat{\bbeta}_j) \bigg\}
\end{equation}
where $\Sigma$ is the corresponding limiting covariance matrix. We have maintained the standard likelihood notational convention by writing it as a function of parameters, given data. 
This is in sharp contrast to the majority of ABC type methods, where such data summaries are typically chosen in an ad-hoc fashion. Unfortunately, the Pitman-Koopman-Darmois (PKD) theorem essentially guarantees the failure of ABC type methods, since it states  that the existence of  finite dimensional sufficient statistics is  a unique property of the exponential family. \gr {The implications of PKD are thus  quite disappointing in the context of ABC methods, where typically  the interesting (non-analytic) likelihoods are outside of the  exponential family. Consequently,  ABC methods based on finite dimensional summary statistics are typically guaranteed to suffer  information loss vis a vis  exact likelihood inference.} 

Consider the $j^{th}$ trajectory and the $(dm_j) \times 1$ vector of stacked species concentrations, $\bC_j=\textrm{vec}(\bX_j)/n$; i.e., by stacking the concentration vectors at each timepoint. The central limit theorem for DDMJP then gives $\sqrt{n}(\bC_j-\bc^\kappa) \Rightarrow \mathcal{N}(0,\Sigma_\kappa)$, so that the trajectory data likelihood, \gr{ $L(\bk|D)=L(\bk|\textrm{vec}(\bX_j)/n)$ converges as $n\to \infty$ to  the Gaussian likelihood $(2 \pi)^{-dm_j/2}|\Sigma_\kappa/n|^{-1/2} \exp \{-1/2(\bC_j-\bc^\kappa)^\top (\Sigma_\kappa/n)^{-1} (\bC_j-\bc^\kappa)\}$. }The law of large numbers and consistency of $\hat{\beta}_j$ imply that $\bC_j \approx \bc^{\hat{\beta}_j}$, so that 
\begin{equation*} L(\bk|\textrm{vec}(\bX_j)/n) \approx (2 \pi)^{-dm_j/2}|\Sigma_\kappa/n|^{-1/2} \exp \{-1/2(\bc^{\hat{\beta}_j}-\bc^\kappa)^\top (\Sigma_\kappa/n)^{-1} (\bc^{\hat{\beta}_j}-\bc^\kappa)\}. \end{equation*} \gr  {The  above  approximation is  seen to hint at   a notion of asymptotic sufficiency (AS) in the sense of \cite{le1956asymptotic}. AS is  essentially an asymptotic Neyman-Fisher factorization, and it implies that, at least asymptotically, the chosen statistics contain meaningful information, thus offering some notion of  efficiency, even in light of the PKD theorem.  However, establishing  AS property more formally requires careful analysis   of  specific inferential problems on  case by case basis.  
For  further discussion, see, for instance,   \cite{frazier2016}.   }
%

It is important to note that in the above arguments, $\Sigma_\kappa$ is the process covariance matrix for the stacked concentration vector in the approximate data likelihood and not the covariance of the summary statistics ($\Sigma$) in the synthetic likelihood. In fact, since $\Sigma$ also depends on $\bk$, its usage in the likelihood  function above would break down  the  conjugacy  and  the efficient MCMC via Gibbs sampling would be no longer available. Thus, when we have a single trajectory, as in the Eyam plague example below, we use the asymptotic covariance matrix of the summary statistic, $\hat{\Sigma}_{\hat{\beta}_j}$ for $\hat{\beta}_j$, to form the synthetic likelihood, although we have suppressed the explicit subscript notation for simplicity. When multiple trajectories are available, we   additionally  have the ability to assess the between trajectory variation of the summary statistics.  See below for details. 

The transition from the original likelihood to the synthetic likelihood shifts our analysis into the setting of a classical  linear model. Thus, this approach is vastly different from most of the currently used  ones. Specifically, methods based on the reaction rate equation use ODE solutions as the means of corresponding Gaussian likelihoods \citep{girolami2008bayesian},  for instance,  those based on the diffusion approximation and the LNA  \citep{komorowski2009bayesian,fearnhead2014inference}. These approaches, while being approximations, still face serious computational challenges. The main bottleneck for inference, particularly Bayesian one, in these models is non-conjugacy,  since the ODE means are not linear in the parameters. As such, each iteration of MCMC requires solving complex systems of nonautonomous ODEs at each proposal value, like in \cite{girolami2008bayesian}. This can make tuning proposal distributions with good acceptance properties difficult, leading to chains with poor mixing. The ABC methods do not fare much better, since they require summary statistic, distance measure, and tolerance selection that are often ad-hoc. These problems  severely limit the applicability and scalability of the current approximate procedures. In contrast, our synthetic likelihood approach circumvents the need to choose distance measures and tolerance levels \gr{by using data summaries that are well understood, and allow for their analysis via standard MCMC. It also only involves solving ODE systems once (to compute initial summary statistics and covariances) thus avoiding the iterative usage of  the ODE solver. Finally, the synthetic likelihood  form leads itself to the efficient formulation of the MCMC computation steps via a Metropolis-within-Gibb's procedure.}

\subsection{Prior Specification} 
The Bayesian approach to network estimation can be addressed by using specialized priors that allow coefficients to be in the model or out of the model during iterations of MCMC, leading to positive  \gr{posterior probabilities of zero value.}Various mixtures of mutually singular distributions, each dominated by $\sigma$-finite measures are a natural choice. Here, we assume a discrete mixture of the point mass at zero, $\delta_0$, and a continuous distribution, $F$, supported on the positive reals and dominated by the Lebesgue measure  $\Pi=(1-\omega)\delta_0+\omega F$. Restricting the support of $F$ to positive reals is necessary to convey the fact that kinetic parameters are non-negative. \gr {It was shown in \cite{gottardo2008markov} that  when the  prior probability of non-zero rate for  reaction $k$ is  $\omega_k$,   the corresponding density for $\kappa_k$ is of the form}
\begin{equation}\label{eq:SLprior}
\pi(\kappa_k):=\frac{d \Pi}{d(\delta_0+\mu)}=(1-\omega_k)\mathbb{I}_0(\kappa_k)+\omega_k f(\kappa_k) \mathbb{I}_{\mathbb{R}^+}(\kappa_k) \textrm{   a.e   }\delta_0+\mu
\end{equation}
where $\mathbb{R}^+=\{x:x>0\}$. Priors of the form \eqref{eq:SLprior} lead to minimax rates of estimation and posterior contraction on sparse sets, provided the tails of $F$ are exponential or heavier (see, e.g., \cite{johnstone2004needles, castillo2012needles}). Thus, priors of the form \eqref{eq:SLprior} are optimal under certain criteria and are usually considered the theoretical gold standard for variable selection in the Bayesian setting. To that end, we assume that $f(\kappa_k|\lambda_k)=\lambda_k \exp\{- \lambda_k \kappa_k \} \mathbb{I}_{\mathbb{R}^+}(\kappa_k)$, which is the exponential, or one-sided Laplace,  distribution. Our choice of the exponential is motivated by several important factors: it naturally restricts the support of $\bk$ to positive reals, it satisfies the tail requirements mentioned previously, and from the information theoretic  perspective  it is  the maximum entropy prior with mean $\frac{1}{\lambda_k}$ \gr{(see \citep{robert2007bayesian}, Chapter 3)}. It may also be rewritten hierarchically as $f_1(\kappa_k|\tau_k)=\sqrt{\frac{2 }{\pi \tau_k}}\exp\{-  \kappa_k^2/2\tau_k \} \mathbb{I}_{\mathbb{R}^+}(\kappa_k)$, with $f_2(\tau_k|\lambda_k)=\frac{\lambda_k^2}{2}\exp \{-\lambda_k^2 \tau_k/2 \}$, via the identity
$\int_0^\infty f_1(\kappa_k|\tau_k) f_2(\tau_k|\lambda_k) d\tau_k=f(\kappa_k|\lambda_k)$
\citep{west1987scale,andrews1974scale}.  By writing the prior as a scaled mixture of truncated normals, the form of $\bk$'s full conditionals becomes analytically tractable, as is demonstrated in the Supplementary Material. This gives significant computational advantage to our framework, since priors like \eqref{eq:SLprior} often lead to non-conjugacy, in which case full Metropolis-Hastings (MH) is required for sampling. Although adaptive MH-based MCMC algorithms have been designed to produce chains with desirable acceptance rates \citep{ji2013adaptive}, when the likelihood is complicated or not analytic, as in the present situation, such tuning is not straightforward. Thus when appropriate tuning cannot be done, the resulting chains may exhibit poor mixing and require extremely long run times for sufficient exploration of parameter space. In what follows, we detail the Gibbs sampling procedure, which performs tuning automatically, for posterior sampling under our hierarchical model. Our primary interest is in obtaining the posterior probability that reactions are true, which allows one to infer the reaction network structure. In order to accommodate varying experimental conditions, such as differences in measurement error or experiments with data collected at different timepoints, we place a Wishart prior on the covariance matrix,  $\Sigma | \Psi \sim \mathcal{W}(v,\Psi)$. This is in contrast to the common assumption of the inverse-Wishart, see \cite{daniels1999nonconjugate,bouriga2013estimation,alvarez2014bayesian} for \gr{some} examples. Our approach is similar to  \cite{chung2015weakly}, in that we assume a Wishart (not inverse-Wishart) prior for $\Sigma$ that leads to desirable modal properties of the posterior. We select the hyperparameter $\Psi$ to be the empirical covariance of $\beta$, and when this estimate is not full rank we add a small regularization term, $10^{-5}$, to its diagonals.
The hierarchical model under consideration is then
\begin{align}\label{eq:Model}
SL(\bk,\Sigma|D)=&\prod_j^N SL_j(\bk,\Sigma|\hat{\bbeta}_j) \nonumber \\
\pi_1(\bk|\bl)=&\prod_{k}\bigg((1-\omega_k)\mathbb{I}_0(\kappa_k)+\omega_k \lambda_k \exp\{- \lambda_k \kappa_k \} \mathbb{I}_{\mathbb{R}^+}(\kappa_k)\bigg) \nonumber \\
\pi_2(\Sigma|\Psi)=&\frac{|\Psi|^{-v/2}|\Sigma|^{(v-d-1)/2}\exp\{-\frac{1}{2}\mathrm{tr}\Psi^{-1} \Sigma \}}{2^{vd/2}\Gamma_d(\frac{v}{2})} 
\end{align}
where $j=1,...,N$ indexes the $N$ independent trajectories of the process. The model contains a covariance term, $\Sigma$, and this parameter may represent intrinsic stochastic noise, as well as measurement error which will dominate in the large volume limit. This parameter is not of particular interest for network or kinetic rate estimation, and a clear advantage of the  Bayesian framework is the ability to marginalize this nuisance parameter \gr{out of}  the posterior. Importantly, \gr {we show }  that the marginal synthetic posterior distribution, $\pi (\bk|D)$ is unimodal when $v \ge N+d+1$ and $\lambda_k=\frac{1-\omega_k}{\omega_k}$ (see Supplementary Material). This is a key property of the proposed method that  not only guarantees identifiability of the reaction network but  also  \gr{contributes to the observed  rapid mixing of the MCMC procedure}.

\subsection{ Posterior Computation} 
Here we describe the algorithm to efficiently sample from the posterior distribution with a Metropolis-within-Gibbs sampler. To simplify notation,  define $U:=nNQ^\top \Sigma^{-1} Q$ and $S:=nQ^\top \Sigma^{-1} \sum_{j=1}^N\hat{\bbeta}_j$. 
The term $\omega_k$ is the prior probability that reaction channel $k$ is true \gr{(non-zero)}. The  posterior computation can then be performed with the following steps. \\
{\bf Algorithm 1.}\\
 Step 1. For $k=1,...,r$, compute $\omega_k^\ast=\frac{\omega_k}{(1-\omega_k)/M_k+\omega_k}$ where
 \begin{equation}
 M_k=2  \sqrt{\frac{1}{\tau_k^2(u_{kk}+1/\tau_k^2)}}\exp\bigg\{ \frac{(s_k-\sum_{i \ne k}u_{ik}\kappa_i)^2}{2(u_{kk}+1/\tau_k^2)} \bigg\} \bigg (1-\Phi \Big(0,\frac{(s_k-\sum_{i \ne k}u_{ik}\kappa_i)}{(u_{kk}+1/\tau_k^2)},(u_{kk}+1/\tau_k^2)^{-1} \Big) \bigg).
\end{equation}
With probability $\omega_k^\ast$, sample $\kappa_k$ from the truncated Gaussian and then $\frac{1}{\tau_k^2}$ from inverse Gaussian, 
\begin{align*}
\kappa_k \sim& \mathcal{N}\bigg(\big(s_k-\sum_{i \ne k}u_{ik} \kappa_i\big)/(u_{kk}+1/\tau_k^2),(u_{kk}+1/\tau_k^2)^{-1}\bigg)\mathbb{I}_{\mathbb{R}^+}(\kappa_k) \\
\frac{1}{\tau_k^2} \sim & I\mathcal{G}(\frac{\lambda_k}{\kappa_k},\lambda_k^2)
\end{align*}
Else, set $\kappa_k=0$ with probability $1-\omega_k^\ast$.
\\
Step 2. Given the current sample $(\bk, \Sigma)$, propose $\Sigma^\ast$ from Wishart, $\Sigma^\ast \sim \mathcal{W}\big(v^\prime, \Sigma \big)$ \\ 
Step 3. Accept $\Sigma^\ast$ with probability $\min \bigg \{1, \frac{\pi(\bk,\Sigma^\ast|D)\mathcal{W}(\Sigma|v^\prime,\Sigma^\ast)}{\pi(\bk,\Sigma|D)\mathcal{W}(\Sigma^\ast|v^\prime,\Sigma)} \bigg \}$ \\
Step 4. Recompute $U:=nNQ^\top \Sigma^{-1} Q$ and $S:=nQ^\top \Sigma^{-1} \sum_{j=1}^N\hat{\bbeta}_j$. Return to step 1. \\
\\
In the above notation $u_{ij}=(U)_{ij}$, $s_k$ is the $k^{th}$ element of $S$, $\mathcal{N}(a,b)$ is a Gaussian random variate with mean $a$ and variance $b$, $\mathcal{W}(\Sigma|v,V)$ is the Wishart density evaluated at $\Sigma$ with scale matrix $V$ and degrees of freedom $v$, $I\mathcal{G}(a,b)$ is an inverse Gaussian random variate with mean $a$ and scale $b$, and $\Phi(x,a,b)$ is the Gaussian cumulative distribution function with mean $a$ and variance $b$ evaluated at $x$. We have found that a proposal degrees of freedom, $v^\prime=n$, gives relatively good acceptance rates, between $15-30$\% in our empirical studies. Derivations of the full conditionals may be found in the Supplementary Material. Hence, sampling from the full conditional of $\bk$ is done by sampling from each $\kappa_k$'s second mixture with probability $\omega_k^\ast$ and from the degenerate component with corresponding probability $1-\omega_k^\ast$. Expressions for the individual parameters' and weights' full conditionals, $\kappa_k|\dots$ and $\omega_k^\ast$, allow for sampling from the target distribution by local parameter-wise updates. Although global moves can lead to optimal acceptance rates, tuning proposals that must be absolutely continuous with respect to measures like $(\delta_0 +\mu)^r$ is not straightforward, and even less so for likelihood-free methods. Additionally, the scheme allows inference about posterior reaction probabilities to be improved via Rao-Blackwellization \citep{gottardo2008markov}.  In the remaining sections we illustrate the usage and performance of Algorithm~1 with both simulated and real data examples.  

\section{Simulation Study}
To illustrate network topology estimation using the proposed synthetic likelihood approach, we consider a molecular reaction network partially motivated by the heat shock response. Heat shock transcription factors and protein chaperones are critical to ensure proteins fold into specific three-dimensional structures. Newly formed proteins and proteins within cells that have been challenged with damage risk protein misfoldings that may effect their functional activity \cite{hartl2011molecular}. Accumulation of such toxic species (misfolded proteins) has been implicated in the progression of certain neurodegenerative diseases \cite{neef2011heat,hartl2011molecular}, and has lead to research into development of theraputic targets that restore proteostasis \cite{calamini2012small}. Hence, modeling the cells ability to employ this chaperone machinery to acheive proteostasis may reveal theraputic targets. As a toy {\it in silico}, we consider the following reaction network that has transcriptional and chaperone components, along with redundant reactions, to compare the proposed methodology with existing ones via simulation.  
\begin{align}\label{eq:sim}
\phi \rightarrow^{\kappa_1} & P_1 \hspace{10mm } \phi \rightarrow^{\kappa_2} P_2 \hspace{10mm} P_1 \rightarrow^{\kappa_3} R_1 \hspace{10mm} P_2 \rightarrow^{\kappa_4}  R_1 \nonumber \\
P_1 \rightarrow^{\kappa_5}  P_1 + P_2&  \hspace{10mm} P_2 \rightarrow^{\kappa_6}  P_1+P_2 \hspace{10mm} R_1 \rightarrow^{\kappa_7} P_2 \hspace{10mm}R_1 \rightarrow^{\kappa_8} 2R_1\\
R_1+P_2 \rightarrow^{\kappa_9} \phi  &   \hspace{10mm} R_1 \rightarrow^{\kappa_{10}}   \phi \hspace{10mm} P_1 \rightarrow^{\kappa_{11}}  \phi \hspace{10mm} P_2 \rightarrow^{\kappa_{12}} \nonumber \phi 
\end{align}
Here $P_1$ and $P_2$ are representing proteins and $R_1$ represents gene RNA expression, with the transitions in/out of $\phi$ indicating loss/creation of a molecule. For the system of reactions \eqref{eq:sim} the mass action ODE \eqref{eq:rre} parameterized by $\bbeta$ specializes to  
\begin{align}\label{eq:ODE}
\frac{dP_1}{dt}=&\beta_1-\beta_2P_1+\beta_3P_2 \nonumber \\ 
\frac{dP_2}{dt}=&\beta_4+\beta_5R_1+\beta_6P_1-\beta_7P_2-\beta_8R_1P_2 \\
\frac{dR_1}{dt}=&-\beta_9R_1+\beta_{10}P_1+\beta_{11}P_2-\beta_8R_1P_2 .\nonumber
\end{align}
 We note that in this particular case $\bbeta=Q\bk$,  where 

\[ Q=\begin{pmatrix} 
1 & 0 & 0 & 0 & 0 & 0 & 0 & 0 & 0 & 0 & 0 & 0   \\ 
0 & 0 & 1 & 0 & 0 & 0 & 0 & 0 & 0 & 0 & 1 & 0    \\  
0 & 0 & 0 & 0 & 0 & 1 & 0 & 0 & 0 & 0 &0 & 0      \\ 
0 & 1 & 0 & 0 & 0 & 0 & 0 & 0 & 0 & 0 & 0 & 0      \\ 
0 & 0 & 0 & 0 & 0 & 0 & 1 & 0 & 0 & 0 & 0 & 0     \\ 
0 & 0 & 0 & 0 & 1 & 0 & 0 & 0 & 0 & 0 & 0 & 0     \\ 
0 & 0 &0 & 1 & 0 & 0 & 0 & 0 & 0 & 0 & 0 & 1       \\ 
0 & 0 & 0 & 0 & 0 & 0 & 0 & 0 & 1 & 0 & 0 & 0      \\ 
0 & 0 & 0 & 0 & 0 & 0 & 1 & -1 & 0 & 1 & 0 & 0      \\ 
0 & 0 & 1 & 0 & 0 & 0 & 0 & 0 & 0 & 0 & 0 & 0       \\ 
0 & 0 & 0 & 1 & 0 & 0 & 0 & 0 & 0 & 0 & 0 & 0  \\
\end{pmatrix} \]

In our present setting $s=3$, $d=11$ and $r=12$. For our simulation study we generated trajectories from the pure jump process of the system of reactions  \eqref{eq:sim} via Gillespie's algorithm (see, e.g., \citep{van1992stochastic}) with parameters $\boldsymbol{\kappa}=[1,0,1,0,0,1,0,0,0.5,1,1,1]^\top$ and initial molecular copy numbers of 50 for each of the three species. Note that under this set of kinetic parameters $P_1$ enters the system from an external source and acts as a transcription factor for $R_1$ and a chaperone for $P_2$, which we model by reactions  1, 3, and 6  (these are labeled by their respective $\kappa$ subscripts in \eqref{eq:sim}). $P_2$ acts as a suppressor of the transcription of $R_1$ through reaction 9 and all species have a natural degradation rate through reactions 10, 11, and 12 respectively. All others reactions are superfluous. 

We calculated the required LSE and MEF -based statistics by fitting the mass action ODE in \eqref{eq:ODE} to $N=1,2,3, \textrm{ and } 5$ simulated stochastic trajectories from \eqref{eq:sim}. We set the degrees of freedom hyperparameter to $v=N+d+1$, $\omega_k=0.5$ for equal \textit{a priori} probability that a reaction channel is true or false, and $\lambda_k=\frac{1-\omega_k}{\omega_k}=1$, which in combination with $v=N+d+1$ guarantees a unimodal posterior. For comparison, we perform analysis using the adaptive MCMC routine of \cite{vihola2012robust} with the LNA likelihood approximation and uniform priors on the logarithm of parameter values \cite{finkenstadt2013quantifying}. Additionally, we implemented the particle filtering routine of \cite{Golightly:2011aa}, which computes unbiased likelihood estimates within MCMC using 100 particles generated via Gillespie's algorithm and assigned uniform priors on the logarithm of parameter values. Tables \ref{table:posteriorlinksLSE} and \ref{table:posteriorlinksMEF} contain the posterior median estimates from chains of 50,000 MCMC samples from the MEF-based and LSE-based synthetic likelihood method. For the class of point mass mixture priors, the posterior median has been proven to be a legitimate thresholding rule, see \cite{johnstone2004needles}, so it may be used for both variable selection and estimation simultaneously in our setup. Tables \ref{table:posteriorLNA} and \ref{table:posteriorPF} give posterior means from LNA analysis with 50,000 samples and 10,000 samples from the particle marginal Metropolis-Hastings algorithm respectively. All algorithms are coded in R and run on a personal desktop computer with 2.7 GHz clock speed.

\begin{table} 
\caption{{\label{table:posteriorlinksLSE} Posterior median of $\kappa_k$ from 50,000 MCMC samples for  $N=1,2,3,5$ trajectories using LSE.}}
\centering
\begin{tabular}
{c c c c c c c c c c c c c c c} \hline 
$\kappa$ & $\kappa_1$ & $\kappa_{2}$ & $\kappa_{3}$ & $\kappa_{4}$ & $\kappa_{5}$ & $\kappa_{6}$ & $\kappa_{7}$ & $\kappa_{8}$ & $\kappa_{9}$ & $\kappa_{10}$ & $\kappa_{11}$ & $\kappa_{12}$ & Time \\ \hline
Truth&1&0&1&0&0&1&0&0&0.5&1&1&1&Seconds \\
$N=1$ & 0.38 & 0.00 & 0.63 & 1.12 & 0.00 & 8.56 & 0.00 & 0.00 & 1.96 & 0.47 & 0.21 & 0.00&47.33  \\
$N=2$ & 0.78 & 0.00 & 0.70 & 1.03 & 0.00 & 4.34 & 0.00 & 0.00 & 1.24 & 0.60 & 0.89 & 0.00&49.91  \\
$N=3$ &0.60&0.00&0.86&0.68&0.00&6.61&0.00&0.00&1.57&0.83&0.38&0.29&50.20 \\
$N=5$ &0.74&0.00&0.92&0.41&0.00&5.50&0.00&0.00&1.14&0.88&0.58&0.58&52.23 \\
\end{tabular}
\end{table}
\begin{table}\caption{\label{table:posteriorlinksMEF}Posterior median of $\kappa_k$  from 50,000 MCMC samples for $N=1,2,3,5$ trajectories using MEF.}
\centering
\begin{tabular}
{c c c c c c c c c c c c c c c} \hline 
$\kappa$ & $\kappa_1$ & $\kappa_{2}$ & $\kappa_{3}$ & $\kappa_{4}$ & $\kappa_{5}$ & $\kappa_{6}$ & $\kappa_{7}$ & $\kappa_{8}$ & $\kappa_{9}$ & $\kappa_{10}$ & $\kappa_{11}$ & $\kappa_{12}$ & Time\\ \hline
Truth&1&0&1&0&0&1&0&0&0.5&1&1&1&Seconds \\
$N=1$ & 1.10 & 0.00 & 0.75 & 1.89 & 0.00 & 1.09 & 0.00 & 0.01 & 0.93 & 0.70 & 1.31 & 2.84&46.97  \\
$N=2$ & 1.04 & 0.00 & 0.85 & 0.95 & 0.00 & 1.04 & 0.00 & 0.00 & 0.73 & 0.82 & 1.13 & 2.13&53.17  \\
$N=3$ &1.02&0.00&0.90&0.63&0.00&1.02&0.00&0.00&0.67&0.88&1.08&1.82&52.55 \\
$N=5$ &1.07&0.00&0.93&0.39&0.00&1.04&0.00&0.07&0.60&1.64&1.15&1.49&53.44 \\
\end{tabular}
\end{table}

\begin{table}\caption{Posterior mean of $\kappa_k$ from 50,000 MCMC samples for $N=1,2,3,5$ trajectories using LNA.}
\centering
\begin{tabular}
{c  c c c c c c c c c c c c c c} \hline 
$\kappa$ & $\kappa_1$ & $\kappa_{2}$ & $\kappa_{3}$ & $\kappa_{4}$ & $\kappa_{5}$ & $\kappa_{6}$ & $\kappa_{7}$ & $\kappa_{8}$ & $\kappa_{9}$ & $\kappa_{10}$ & $\kappa_{11}$ & $\kappa_{12}$ & Time\\ \hline
Truth&1&0&1&0&0&1&0&0&0.5&1&1&1&Seconds \\
$N=1$ & 0.95 & 0.00 & 0.79 & 1.52 & 0.00 & 0.40 & 0.00 & 1.11& 0.02 & 1.85 & 1.02 & 0.01 & 998356.82  \\
$N=2$ & 1.06 & 0.00 & 0.84 & 1.33 & 0.00 & 0.12 & 0.00 & 0.00 & 0.25 & 0.82 & 1.22 &  0.01 & 1002223.98\\
$N=3$ &1.01&0.00&0.93&1.30&0.00&0.10&0.00&0.00&0.17&0.96&1.02&0.00&998999.92 \\
$N=5$ &1.05&0.00&0.90&1.19&0.00&0.54&0.00&0.00&0.29&0.91&1.14&0.01&998811.77 \\
\end{tabular}
\label{table:posteriorLNA}
\end{table}

\begin{table}\caption{Posterior mean of $\kappa_k$ from 10,000 MCMC samples for $N=1,2,3,5$ trajectories using particle filtering.}
\centering
\begin{tabular}
{c c c c c c c c c c c c c c c} \hline 
$\kappa$ & $\kappa_1$ & $\kappa_{2}$ & $\kappa_{3}$ & $\kappa_{4}$ & $\kappa_{5}$ & $\kappa_{6}$ & $\kappa_{7}$ & $\kappa_{8}$ & $\kappa_{9}$ & $\kappa_{10}$ & $\kappa_{11}$ & $\kappa_{12}$ & Time\\ \hline
Truth&1&0&1&0&0&1&0&0&0.5&1&1&1&Seconds \\
$N=1$ & 0.62 & 0.01 & 0.83 & 0.71 & 0.01 & 6.05 & 0.01 & 0.15& 3.01 & 0.89 & 0.43 & 1.08 & 1518414.34  \\
$N=2$ & 1.02 & 0.01 & 0.94 & 0.05 & 0.02 & 0.24 & 0.02 & 0.05 & 7.23 & 0.91 & 1.04 & 1.56 & 1546878.77\\
$N=3$ &0.95&0.00&0.99&0.31&0.00&0.29&0.01&0.08&3.54&1.07&0.83&0.93&1546693.39 \\
$N=5$ &0.63&0.00&0.90&0.07&0.01&7.30&0.01&0.05&3.30&0.87&0.40&1.13&1517892.47 \\
\end{tabular}
\label{table:posteriorPF}
\end{table}

The results in Tables \ref{table:posteriorlinksLSE} and \ref{table:posteriorlinksMEF} indicate that both synthetic likelihood methods have performed reasonably well in the example, with the MEF-based synthetic likelihood analysis performing the best. The improvement of the MEF-based synthetic likelihood over the LSE-based one is likely due to a generally better efficiency of MEF over that of LSE. There is also apparently some degree of bias in the LSE estimate of $\kappa_6$, even with increasing $N$. MEF-based likelihood analysis assigned to the true reactions non-zero posterior medians for all sample sizes, while assigning zero posterior medians to nearly all the false reactions for two or more trajectories, with the exception of $\kappa_4$. MEF-based analysis performed better than LSE-based analysis for all reactions and for each number of trajectories. Both LSE and MEF-based analysis gave high posterior probability to the protein $P_1$ being a transcription factor for $R_1$(reaction 3) as well as $P_2$ acting as a suppressor of transcription of $R_1$ (reaction 9), which was indeed the case in the  simulation. Since in actual experiments such reactions often indicate drug targets, the fact that our method was able to correctly identify them is of practical relevance. The results for both the LNA and particle filtering, located in Tables \ref{table:posteriorLNA} and \ref{table:posteriorPF}, show that their corresponding run times are already unacceptable in this moderate sized system. Indeed, we found that obtaining just 10,000 samples from the posterior distribution with the particle marginal Metropolis-Hastings implementation required more than two weeks of CPU time. Collecting 50,000 samples from the posterior distribution using the LNA, with adaptive MCMC for optimal acceptance rates of 0.234, took approximately 12 days of CPU time. The bottleneck of computation for the particle filtering is that unbiased likelihood estimates require sampling many trajectories, in our case 100, for each likelihood evaluation. For the LNA based analysis, each likelihood evaluation requires solving a system of non-autonomous ODEs. The particular examples in Tables~\ref{table:posteriorLNA} and \ref{table:posteriorPF} illustrate the generally accepted view that, at least until now, most of the current methods that rely on detailed system modeling do not scale well. Not only did the methods perform poorly in terms of long run times, they also produced estimates that appear biased away from the true values of certain parameters, even as the number of trajectories increases. For the particle filtering routine, estimates for $\kappa_6$ and $\kappa_9$ have a high degree of bias, whereas the LNA-based analysis appears to incorrectly infer $\kappa_4$ as true (non zero) and $\kappa_{12}$ as false (zero), even with  all  5 trajectories data provided. The synthetic likelihood methods perform better, in terms of inference and computation, by projecting a high dimensional noisy trajectory into a lower dimensional statistic that captures the important dynamical information with less noise. It appears that for the full Metropolis type methods, the variable selection priors, like the point mass ones used in our synthetic likelihood methods, would likely pose even greater computational challenges than the continuous priors applied in our examples here since the routines for tuning the necessary proposals are not straightforward for the LNA, the particle filter, or any Metropolis type sampling with intractable likelihoods.   

\begin{figure}[htbp]
   \centering
   \includegraphics[width=6in,height=5in]{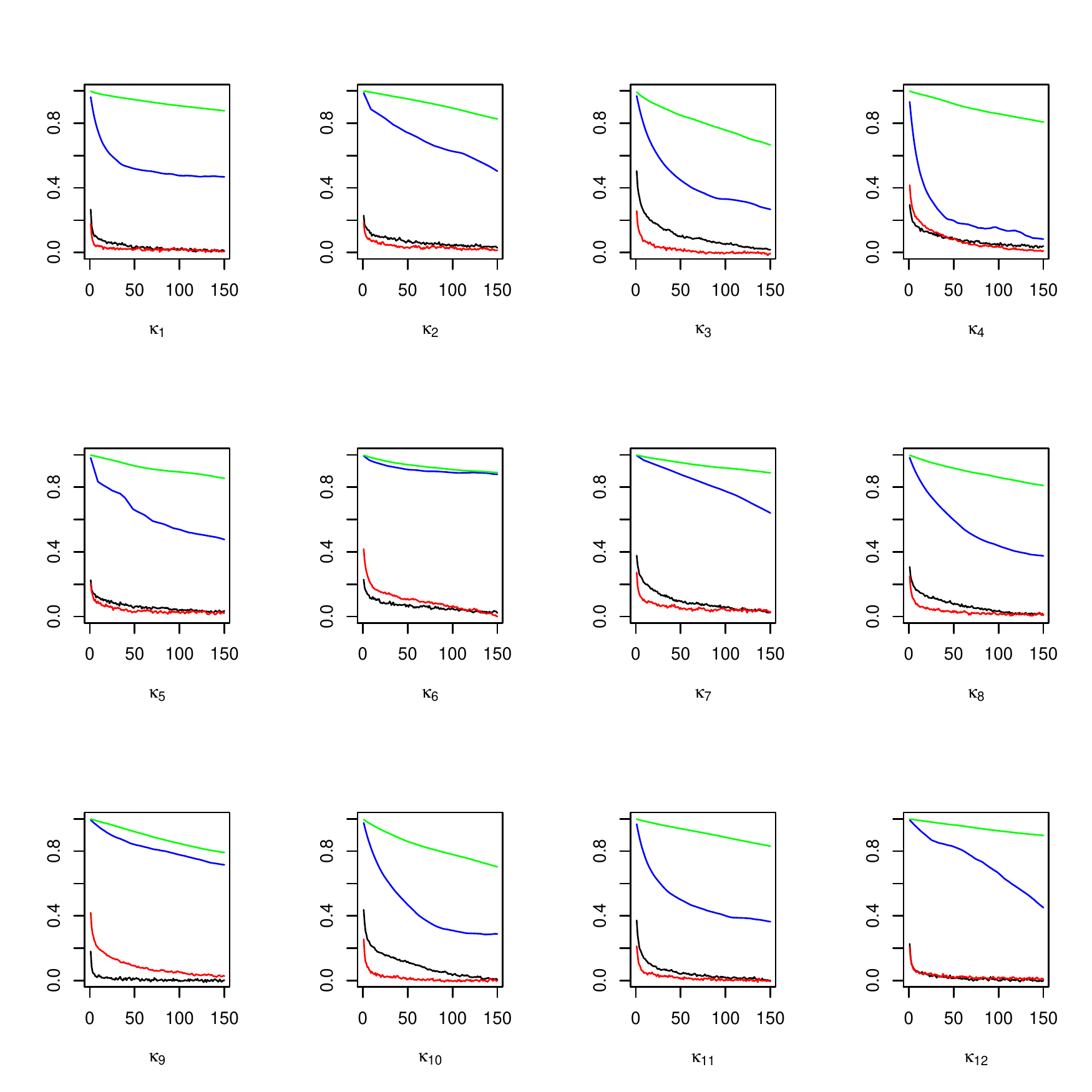} 
   \caption{{Auto-correlation plots of the output from MEF-based synthetic likelihood (black), LSE-based synthetic likelihood (red), LNA (blue), and particle marginal Metropolis-Hastings (green) for $N=1$.}}
   \label{fig:autocorr}
\end{figure}

The plots in Figure \ref{fig:autocorr} indicate that the chains resulting from the LNA and particle filtering have a much higher degree of auto-correlation as compared to the synthetic likelihood methods. Thus, in addition to the increased computational time of each MCMC sample, more samples are required in order for the LNA and particle filter to sufficiently explore parameter space in our current setting. We conclude by these plots that although adaptive MCMC was used, at least for the LNA likelihood approximation, the resulting chains exhibit poor mixing. Although theoretically Metropolis-Hastings type samplers can be tuned to produce optimal acceptance properties, the results here highlight the general difficulty of tuning in the presence of complicated or intractable likelihoods. 


\section{Data Examples}
\subsection{RNA-Seq  Data}
We now compare the performance of our synthetic method to that of the algebraic statistical model (ASM). The method was introduced in \cite{craciun2009algebraic}  to learn biochemical network topology from the empirical patterns of the reaction stoichiometries ($\nu_k ^\prime-\nu_k $). To facilitate the comparison,  we re-analyze a dataset introduced in  \cite{Linder:2013aa}  and  consisting of the longitudinal    RNA-seq measurements from the  retinal tissue in  the zebrafish ({\it Danio rerio}). The study was performed to probe the regenerative properties of the zebrafish retina after  it sustained cell-specific damage. One interest of the study was in analyzing a particular sub-system, consisting of the following species: heat shock protein transcription factor (Hsp70), signal transducer and activator of transcription 3 (Stat3), and the suppressor of cytokine signaling 3 (Socs3). For more details on the experiment, see \cite{Linder:2013aa}. The network of interest has the form
\begin{align}\label{eq:ZebraNet}
\emptyset& \xrightarrow{\kappa_1} Stat3 \hspace{5mm}Stat3 \xrightarrow{\kappa_4} 2Stat3 \nonumber \\
\emptyset& \xrightarrow{\kappa_2} Socs3b \hspace{5mm} Stat3\xrightarrow{\kappa_5} Socs3b \nonumber \\
\emptyset& \xrightarrow{\kappa_3} hsp70 \hspace{5mm} Stat3 \xrightarrow{\kappa_6} hsp70 \nonumber \\
S&tat3+Socs3b \xrightarrow{\kappa_7} Socs3b \nonumber \\
S&tat3 \xrightarrow{\kappa_8} \emptyset \nonumber \\
S&ocs3b\xrightarrow{\kappa_9} \emptyset \nonumber \\
h&sp70 \xrightarrow{\kappa_{10}} \emptyset.  
\end{align}
In the above network, we are especially interested in  the possible activation of the heat shock response via Stat3. The detailed  analysis via ASM based on all 8 trajectories of the experiment  was presented  in \cite{Linder:2013aa}, where the topology of the conic (i.e., single-source)  sub-network in Figure \ref{fig:Stat3Cone} was learned. We may thus compare the proposed synthetic method's results  based on LSE with the results based on ASM for the same dataset. As previously mentioned, the proposed method also allows for computation of posterior probabilities via empricial estimates of the posterior weights; i.e., $P(\kappa_k \ne 0|D)$ since the dominating measure is $(\delta_0+\mu)$ and not merely $\mu$. To this end, we compute and report the posterior probabilities by simulating 50,000 MCMC samples from the model in \eqref{eq:Model} under the same hyperparameter assumptions, as in the previous section. 

\begin{figure}
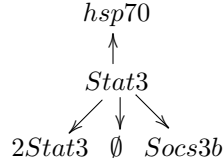
\centering \caption{{Stat3 Conic Network. Note that the only source for 4 different products is Stat3.}} \label{fig:Stat3Cone}
\begin{equation*}
\xy ;<.7pc, 0pc>:
\POS(0,0)*+{Stat3}
\ar@<.3 ex> @{->}        +(-3,-3)*+{2Stat3}
\ar@<.3 ex> @{->}        +(3,-3)*+{Socs3b}
\ar@<.3 ex> @{->}        +(0,3)*+{hsp70}
\ar@<.3 ex> @{->}        +(0,-3)*+{\emptyset}
\endxy
\end{equation*}
\end{figure}

The results in Table \ref{tab:zebraTable} indicate that reactions 4, 5, and 6 are likely true, while reaction 8 may only occur on a much longer time scale. Since both methods produce similar network topologies, we mention some advantages of the proposed model over ASM. While the appeal of the ASM is that it exploits the geometry of the stoichiometric matrix, the proposed method based on synthetic likelihood does so as well, in a sense,  through the entries of $Q$ matrix. A practical limitation of the ASM is that it enforces the cone-wise assumption that exactly $s$ reactions are true, which will typically not be the case. Similarly to the synthetic method,  ASM also tacitly assumes a large volume setting, $(n \rightarrow \infty)$, however, unlike for the synthetic method, the ASM inference problem is only asymptotically $(N \rightarrow \infty)$ well-posed and only on the  set of posterior probabilities  $\omega_k^\ast\in \{ 0,1\}$ for $k=1,\ldots,r$. Thus ASM is strictly a topology learning routine, and not capable of  kinetic parameter estimation. In contrast, (even though we did not present the results in this section for brevity) the parameter estimation may be easily carried out with the proposed synthetic approach by analyzing the posterior distribution and selecting the point estimates of $\kappa_k$, as was done in the previous section. For illustration, we  present the bivariate contour plots of the posterior distribution for the reaction rates from the sub-system of interest in Figure \ref{fig:contour}. Our main observation is that the empirical plots indeed agree with our theoretical  results  on the  unimodality of the posterior distribution. 

\begin{table}\caption{\label{tab:zebraTable}Synthetic likelihood and ASM reaction probabilities}
\centering 
\begin{tabular}{r r r r}
\hline
Reaction Source&Reaction Output&Synthetic Likelihood&ASM \\
\hline 
$Stat3$&$2Stat3$&$0.94$&$1$ \\
$Stat3$&$Socs3b$&$0.97$&$1$ \\
$Stat3$&$hsp70$&$0.99$&$1$ \\ 
$Stat3$&$\emptyset$&$0.34$&$0$ \\
\hline
\end{tabular}
\end{table}
\begin{figure}[htbp]
   \centering
   \includegraphics[width=6in,height=5in]{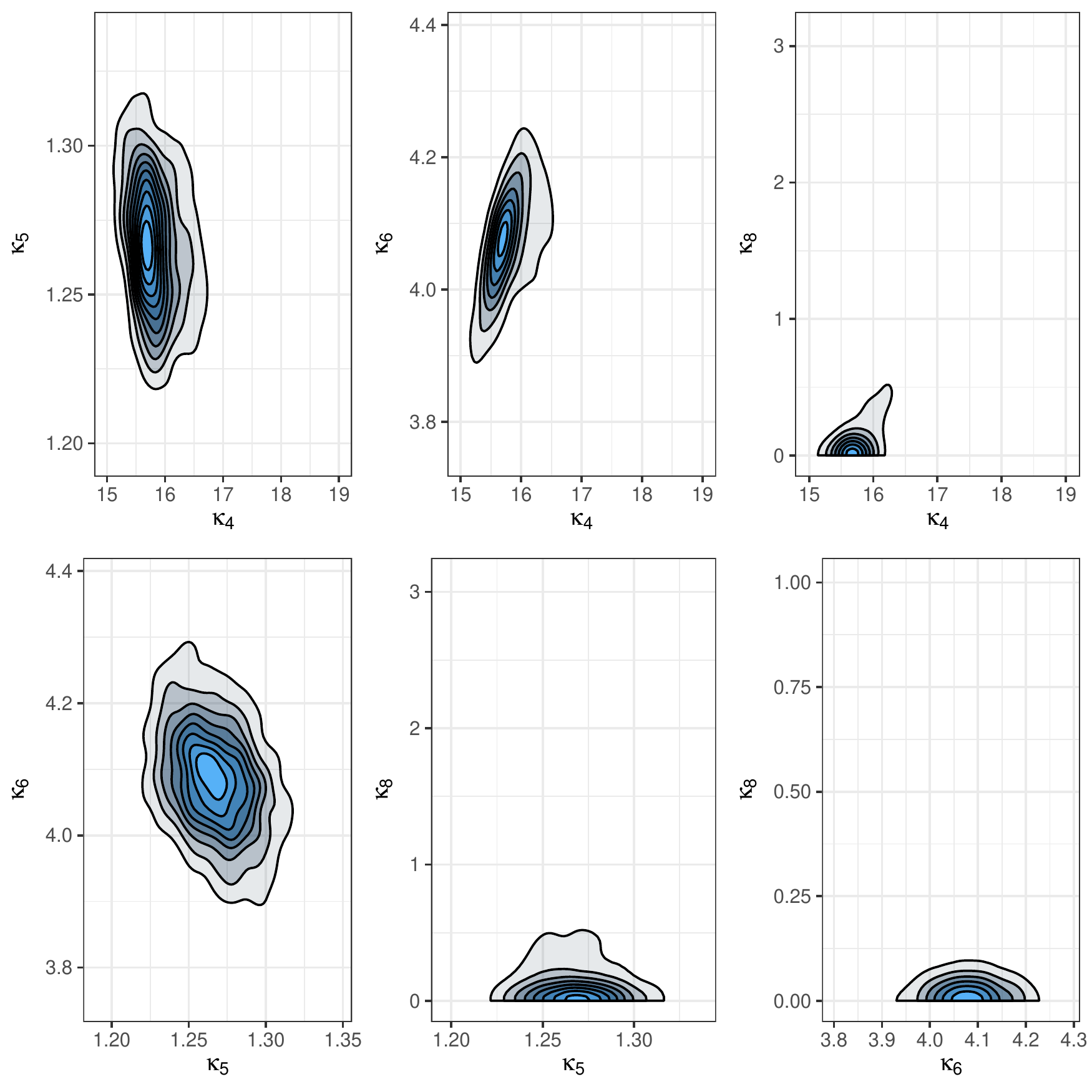} 
   \caption{{Bivariate contour plots of posterior distribution from synthetic likelihood. Empirical density estimates increase from light gray to light blue.}}
   \label{fig:contour}
\end{figure}
\subsection{The Plague at Eyam}
 In the seventeeth century,   following the Great Plague of London the village of Eyam, Derbyshire, England, experienced an outbreak of plague, caused by the bacterium {\it Yersinia pestis}. In this section we analyze data from this outbreak that occurred at Eyam in 1665-1666.  See  also \cite{whittles2016epidemiological, dean2018human}
  for further discussion of the dataset and the relevant historic context.  


Several features about the plague outbreak at Eyam make its study somewhat unique. The first of these features is  that the village rector, a Reverand William Mompesson,  reportededly  convinced the villagers to self-quarantine. Although recent evidence suggests that a few of the wealthy residents may have fled (it is  reported that Mompesson sent his children away before the quarantine), we may effectively treat the plague at Eyam as an  outbreak in  closed population. The names and burial dates of plague victims were recorded by Mompesson. Further, the parish records combined with the hearth tax record for Eyam in 1664 provide detailed information on the villagers, such as their sex, approximate date of birth, date of burial, and household information. This curated version of the Eyam parish register has lead to a newly revised estimate of a total village population of around 700, from an initially reported 350.   
         
As the account goes, a tailor at Eyam received a shipment of cloth from London that was carrying plague infected rat fleas, and the first infected victim is believed to have come in contact with this cloth. As the infected flea's digestive system becomes blocked by the bacterium, the flea vomits into the bite wound, thus transmitting {\it Y. pestis}. This transmission mechanism is now medically confirmed as giving rise to the bubonic form of  plague. On September 7th 1665 the first burial due to plague was for a George Viccars. Over the next nine months, a somewhat constrained outbreak occurred in the Eyam villagers, of which 77 deaths have been attributed to plague. Around mid-May 1666 a second wave of the outbreak began to spread, and during the ensuing months from June 1666 through October 1666 had decimated the village, killing some 257 villagers.         

While the rodent-to-human transmission route via the rat flea is understood to be critical for the initial outbreak dynamics, this particular mode of transmission alone does not fully explain the observed rapidity of the various plague outbreaks throughout Europe. This was also argued,  at least qualitatively,  based on the empirical differences in the early outbreak dynamics compared to the latter months at Eyam, \cite{raggett1982stochastic}. An apparent lack of recorded rat falls (large scale rat deaths) during these outbreaks provides further evidence that additional transmission mechanisms were also critical for disease spread. Rat falls are generally considered necessary to cause sufficient flea jumpings from rat corpses onto humans. While human-to-human contact has been recognized as a component of the plague transmission process, through plague pneumonia and more recently via ectoparasites, such as lice and the human flea, recent analyses suggest that this transmission route may be far more important than previously recognized \citep{whittles2016epidemiological, dean2018human}.

We set out to analyze the Eyam plague data that was reported in \cite{raggett1982stochastic}, which we have augmented to account for the more recent information on the total population reported in \cite{whittles2016epidemiological}.  
      
\tikzset{int/.style={draw, line width = .75mm, minimum size=5em}}

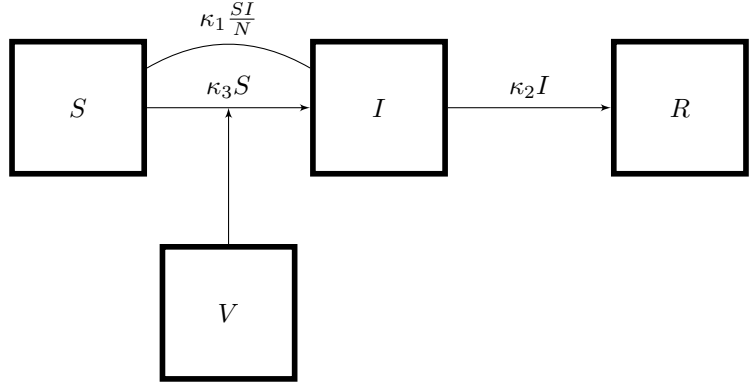
\begin{figure} 
\centering

\begin{tikzpicture}[node distance=4cm,auto,>=latex']
\node [int] (a) {$S$};
\node [int] (b) [right of=a, node distance=4cm]{$I$};
\node [int] (c) [right of=b, node distance=4cm]{$R$};

\draw[->] (a) -- node[name=u]{$\kappa_3 S$} (b);
\draw[->] (b) -- node[name=v] {$\kappa_2 I$} (c);
\node [int] (r) [below of=u, node distance=3cm]{$V$};
\path (a) edge[bend left] node[anchor=north, above]{$\kappa_1 \frac{SI}{N}$} (b);
\draw[->] (r) -- node[name=j] {$$} (u);


\end{tikzpicture}
\caption{Susceptible, Infected, Recovered (SIR) compartmental model for the Eyam plague.}
\end{figure}
Figure 4 illustrates the compartmental SIR model that we consider for analysis of the Eyam data. The $S$, $I$, and $R$ compartments represents the number of susceptibles, infectious, and removed individuals, which we denote at time $t$ by $S(t)$, $I(t)$, and $R(t)$. While we have labeled the compartment $R$ removed as is standard in the typical SIR notation, the Eyam plague was almost universally fatal for infected. There were only three alleged recoveries, of which none were reported in the data sources, so that $R$ effectively represents deaths. We represent the rodent-flea $V$ compartments and its contribution to the infectious pressure on susceptibles as $\kappa_3S$. This is essentially an assumption that the infectious pressure from non-human interaction is proportional to the number of susceptibles. While this may not be a completely accurate description, the outbreak period that we analyze was during early to late summer, leading one to suspect that  the infectious pressure from the $V$ compartment may have been approximately constant. According to the SIR model,  susceptibles ($S$) make infectious contact and then transition to compartment $I$ at rate $\kappa_1SI$. Finally, infected individuals die and transition to compartment $R$ at rate $\kappa_2I$.     
 
The compartment specific prevalence estimates, reported in \cite{raggett1982stochastic}, were updated with the new population total and are displayed in Table \ref{table:Eyam}. We have focused attention on the second phase of the outbreak that occurred in the summer of 1666 to assess the evidence on  whether a particular transmission route was more important than another. Since the Eyam data analysis assumes a closed population one may readily recover the $R$ compartment at time $t$ by the formula $N=S(t)+I(t)+R(t)$, with $N=613$.
 \begin{table} 
\caption{{\label{table:Eyam} Plague data for Eyam 1666}}
\centering
\bgroup
\setlength{\tabcolsep}{4em}
\def\arraystretch{1.5}
\begin{tabular}
{c c c} 
Date 1666& $S(t)$ & $I(t)$  \\ \hline
May 18&612&1 \\
June 18 & 593 & 7   \\
July 19 & 540& 22   \\
August 19 &460&20 \\
September 19 &436&8 \\
October 20 &422&0 
\end{tabular}
\egroup
\end{table}    
The corresponding mass action ODE then for the above system has the form
\begin{align}\label{EyamODE}
\frac{dS}{dt} =& -(\kappa_1I+\kappa_3)S \nonumber \\
\frac{dI}{dt} =&-(\kappa_2-\kappa_1 S)I 
\end{align}
We use the plague data in Table \ref{table:Eyam} to compute the MEF statistic using the approach described above. The MEF statistic is computed by minimizing the weighted sum of squared distances between the observed trajectory and ODE solution of \eqref{EyamODE}, weighted by the asymptotic process covariance at each timepoint. The initial condition is given as the compartment specific prevalence estimates in May 18 1666. For the Eyam data we have the single outbreak trajectory; i.e., no replicates, and we thus use the asymptotic covariance estimate of the MEF statistic for the fixed covariance term to construct the synthetic likelihood. Also, $Q=I_{3 \times 3}$, indicating that the unknown rate parameters are directly related to the summary statistics through the identity matrix. We collected 50,000 MCMC samples via the synthetic likelihood method described above with a burn-in of 5,000. The corresponding posterior medians and 95\% credible intervals for the parameters are $\kappa_1= 5.30\hspace{1mm}(5.22, 5.38)$,  $\kappa_2= 4.22\hspace{1mm} (4.14, 4.30)$, and $\kappa_3= 0 \hspace{1mm}(0, 0.003)$.

We note that the results from our  analysis agree qualitatively with the results in \cite{whittles2016epidemiological, dean2018human} concerning  the role of human-to-human transmission of plague. While we have not made explicit assumptions about what the exact form of  the human-to-human transmission mechanism (i.e, plague pneumonia, ectoparasites or some other form) the data from the latter months of the outbreak at Eyam nonetheless suggest that in our simple modified SIR model human-to-human contact was important. This was already suggested early on from the historical accounts that the plague at Eyam could be transmitted from the cough of a plague victim, suggesting plague pneumonia transmission. Further, the posterior median for $\kappa_3$ of zero indicates the corresponding link in our modified SIR model, that accounts for infectious pressure from the rodent-flea route, could be negligible, at least during the latter part of the outbreak.   

There are several limitations of this analysis that should be noted here. The first is that we have restricted our attention to the latter months of the outbreak at Eyam, during which it was apparent that the dynamics had changed from those of the initial outbreak. By doing this we are potentially missing information about the nature of the initial dynamics, which may point to a different transmission route as being important early on. Indeed, results from \cite{whittles2016epidemiological} suggest that approximately 27\% of infections were caused by rodents and 73\% from human-to-human transmission by using the full outbreak data. This leads to another limitation, in that we have relied on the prevalence estimates reported in \cite{raggett1982stochastic}, updated with new population totals. 
These compartment prevalences were estimated from the historical and death records, so are likely subject to measurement error, which we have not accounted for. Further, we have not used data on household structure and composition, although this is part of planned future work. Finally, while the augmented SIR type model we have used is somewhat similar to the SEIR model used in \cite{whittles2016epidemiological}, it does not consider explicit plague pneumonia vs. ectoparasite driven human-to-human transmission separately, as was done in \cite{dean2018human}. Hence, our analysis only adds to the evidence that some form of human-to-human transmission, which we modeled with a generic SIR framework, was important but does not distinguish between particular forms of this transmission.            
\section{Conclusions} We have described a method which can be used to perform estimation of biochemical networks as often considered in the context of dynamic gene regulatory networks and stochastic epidemic models. 
It is well established that this is a notoriously difficult problem, due to the intractability of the likelihood under partially observed trajectories. The underlying theme in most of the popular approaches in this area is to use likelihood approximations to perform approximate inference, such as in \cite{golightly2005bayesian,girolami2008bayesian,komorowski2009bayesian,golightly2012,finkenstadt2013quantifying,fearnhead2014inference}. While our approach adopts this  theme, it is fundamentally different than the standard approximate and likelihood-free inferential techniques. The most important of these differences is that the data summary statistics used here (LSE or MEF estimates) have properties that are well understood and are directly related to the unknown kinetic parameters. These properties justify, via the asymptotic normality, a parametric form for the synthetic likelihood, in the spirit of \cite{Wood:2010aa}, that is linear in the parameters of interest. Hence, we have demonstrated that projections of the species' trajectories into sets of dynamically informative statistics allows for highly efficient posterior sampling and a procedure that should scale well in large systems.

\section*{ACKNOWLEDGEMENTS}
The first author would like to thank the Mathematical Biosciences Institute (MBI) at Ohio State University, for partially supporting this research through an Early Career Award. MBI receives its funding through the National Science Foundation grant DMS 1440386.
\appendix

\section{Proofs of Propositions}
\begin{prop}
The statistics, $\hat{\bbeta}_j$, of Equations 4 and 5 computed from the $j^{th}$ trajectory of data arising from the DDMJP in Equation 1 are asymptotically sufficient for $\bbeta=Q\bk$ as $n \rightarrow \infty$.
\end{prop}
\begin{proof}
Let $D_j$ be the $j^{th}$ trajectory arising from the DDMJP in Equation 1. Assume that $\bk$, and hence $\bbeta=Q\bk$, is the true parameter and let $\hat{\bbeta}$ be a solution to $G_j(\hat{\bbeta})=0$ in Equation 5, which necessarily satisfies 
\begin{align*}
0=&\sum_{t_{ij}} \partial c^{\hat{\bbeta}}(t_{ij})(Var_{\hat{\bbeta}}(C_n(t_{ij})))^{-1}(C_n(t_{ij})-c^{\hat{\bbeta}}(t_{ij})) \nonumber \\
=&\sum_{t_{ij}} \partial c^{\hat{\bbeta}}(t_{ij})(Var_{\hat{\bbeta}}(C_n(t_{ij})))^{-1}(C_n(t_{ij})-c^{\bbeta}(t_{ij})+c^{\bbeta}(t_{ij})-c^{\hat{\bbeta}}(t_{ij})) \nonumber \\
=&\sum_{t_{ij}} \partial c^{\hat{\bbeta}}(t_{ij})(Var_{\hat{\bbeta}}(C_n(t_{ij})))^{-1}(C_n(t_{ij})-c^{\bbeta}(t_{ij})) \nonumber \\ &+\sum_{t_{ij}} \partial c^{\hat{\bbeta}}(t_{ij})(Var_{\hat{\bbeta}}(C_n(t_{ij})))^{-1}(c^{\bbeta}(t_{ij})-c^{\hat{\bbeta}}(t_{ij}))
\end{align*}
implying that
\begin{equation}\label{eq:proof}
\sum_{t_{ij}} \partial c^{\hat{\bbeta}}(t_{ij})(Var_{\hat{\bbeta}}(C_n(t_{ij})))^{-1}(c^{\hat{\bbeta}}(t_{ij})-c^{\bbeta}(t_{ij}))=\sum_{t_{ij}} \partial c^{\hat{\bbeta}}(t_{ij})(Var_{\hat{\bbeta}}(C_n(t_{ij})))^{-1}(C_n(t_{ij})-c^{\bbeta}(t_{ij}))
\end{equation}
Taylor's expansion of the left hand side of Equation \eqref{eq:proof} about $\bbeta=Q\bk$ and scaling by $\sqrt{n}$ gives
\begin{equation}
\sqrt{n}(\hat{\bbeta}-Q\bk)= (B_{\hat{\bbeta}})^{-1} \sum_{t_{ij}}\partial c^{\hat{\bbeta}}(t_{ij})(Var_{\hat{\bbeta}}(C_n(t_{ij})))^{-1}\sqrt{n}(C_{n}(t_{ij})-c^{\bbeta}(t_{ij}))+o_P(1)
\end{equation}

where $ B_{\hat{\bbeta}}=\sum_{i}\partial c^{\hat{\bbeta}}(t_{i})(Var_{\hat{\bbeta}}(C_n(t_{ij})))^{-1}[\partial c^{\hat{\bbeta}}(t_i)]^\top$ and the higher order terms in the expansion vanish since $\hat{\bbeta}$ is consistent for $\bbeta$, under the regularity conditions in \cite{Remp12}. Consistency and the asymptotic normality of $\sqrt{n}(C_{n}(t_{ij})-c^{\bbeta}(t_{ij}))$, imply that $\sqrt{n}(\hat{\bbeta}-Q\bk) \Rightarrow \mathcal{N}(0,\Sigma)$, where $\Sigma$ is the limiting covariance, see \cite{Linder:2013aa,Linder:2015aa}.  
\end{proof}
\begin{prop}
The point mass mixture prior, $\pi_1(\bk)$, in Equation 9 with \\ $f(\kk)=\lambda_k \exp\{- \lambda_k \kappa_k \} \mathbb{I}_{\mathbb{R}^+}(\kappa_k)$ and $\lambda_k=\frac{1-\omega_k}{\omega_k}$ is logarithmically concave in $\bk$, and hence is unimodal.
\end{prop}
\begin{proof}
We prove this result component-wise and the result for the full vector $\bk$ follows. The exponential density, $f(\kk)=\lambda_k \exp\{-\lambda_k \kk \}$ belongs to the class of log-concave densities; i.e., for any $x, y \in \mathbb{R}^+$ we have $f(\alpha x + (1-\alpha)y) \ge f(x)^{\alpha}f(y)^{1-\alpha}$ for $\alpha \in (0,1)$. Thus, if $x$ and $y$ are both positive we have 
\begin{align*}
\pi_1(\alpha x + (1-\alpha)y)=\omega f(\alpha x + (1-\alpha)y) & \ge \omega f(x)^{\alpha}f(y)^{1-\alpha} \\ & =(\omega f(x))^{\alpha} (\omega f(y))^{1-\alpha} \\ &=
\pi_1(x)^{\alpha} \pi_1(y)^{1-\alpha}. 
\end{align*}
When both $x$ and $y$ are zero we have $\pi(\alpha x + (1-\alpha)y)=\pi(0)=1-\omega = (1-\omega)^{\alpha}(1-\omega)^{1-\alpha}=\pi(0)^{\alpha}\pi(0)^{1-\alpha}$. When only one is zero, say $y$,  then 
\begin{align*}
\pi(\alpha x + (1-\alpha)y)=\omega f(\alpha x + (1-\alpha)y) & \ge \omega f(x)^{\alpha}f(0)^{1-\alpha}\\ & =(\omega f(x))^{\alpha} (\omega f(0))^{1-\alpha} \\ &=
\pi(x)^{\alpha} (\omega \frac{1-\omega}{\omega})^{1-\alpha} \\ &= \pi(x)^{\alpha} \pi(0)^{1-\alpha}.
\end{align*}
Thus, the prior is logarithmically concave in $\bk$ and hence is also unimodal.
\end{proof}
\begin{prop}
If $\lambda_k=\frac{1-\omega_k}{\omega_k}$ and $v \ge N+p+1$, the marginal synthetic posterior distribution from the synthetic likelihood model, $\pi(\bk|D)$, is unimodal.
\end{prop}
\begin{proof}
The synthetic posterior is proportional to 
\begin{align}\label{eq:posterior}
\pi(\bk,\Sigma|D) &\propto SL(\bk,\Sigma|D)\pi_1(\bk|\bl)\pi_2(\Sigma|\Psi) \nonumber \\ & \propto |\Sigma|^{(v-N-d-1)/2} \exp\big\{-\frac{1}{2}\mathrm{tr}(\Psi^{-1} \Sigma) \big\} \exp\big\{-\frac{n}{2}\sum_j (Q \bk-\hat{\beta}_j)^\top \Sigma^{-1}(Q \bk-\hat{\beta}_j)\big\} \nonumber \\  & \times\pi_1(\bk|\bl).
\end{align}
The logarithm of the first term in the product, $(v-N-d-1)/2 log |\Sigma|$, is concave when $v \ge N+d+1$, since $log |\Sigma|$ is concave. Logarithmic concavity of the second term follows from convexity of $tr(\Sigma)$. Consider an individual factor from the third term, $\exp\{-\frac{n}{2}(Q \bk-\hat{\beta}_j)^\top \Sigma^{-1}(Q \bk-\hat{\beta}_j) \} $.
The function $(Q \bk-\hat{\beta}_j)^\top \Sigma^{-1}(Q \bk-\hat{\beta}_j) $ has epigraph $\textrm{epi}\big((Q \bk-\hat{\beta}_j)^\top \Sigma^{-1}(Q \bk-\hat{\beta}_j)\big)=\big\{  \bigg(\big((Q \bk-\hat{\beta}_j), \Sigma \big),t \bigg)| \Sigma \succ 0, (Q \bk-\hat{\beta}_j)^\top \Sigma^{-1} (Q\bk-\hat{\beta}_j) \le t \big\}$. This is equivalent to \\ $\big\{  \bigg(\big((Q \bk-\hat{\beta}_j), \Sigma\big),t\bigg)| \Sigma \succ 0, \left(\begin{array}{cc} 
\Sigma & (Q \bk-\hat{\beta}_j)  \\ (Q \bk-\hat{\beta}_j)^\top & t  \end{array}\right) \succeq 0 \big\}$ via the Schur complement. Since the last condition is a linear matrix inequality in $((Q \bk-\hat{\beta}_j),\Sigma,t)$, $\mathrm{epi}((Q \bk-\hat{\beta}_j)^\top \Sigma^{-1}(Q \bk-\hat{\beta}_j))$ is convex, which implies that $(Q \bk-\hat{\beta}_j)^\top \Sigma^{-1}(Q \bk-\hat{\beta}_j)$ is convex in $((Q \bk-\hat{\beta}_j),\Sigma)$ \cite[Ch. 3]{boyd2004convex}. Thus, $h_j((Q \bk-\hat{\beta}_j),\Sigma)=\exp\{-\frac{n}{2}(Q \bk-\hat{\beta}_j)^\top \Sigma^{-1}(Q \bk-\hat{\beta}_j) \} $ is logarithmically concave in $((Q \bk-\hat{\beta}_j),\Sigma)$. That this implies logarithmic concavity in $(\bk,\Sigma)$ follows by considering $h_j^\ast(\bk,\Sigma):=h_j((Q \bk-\hat{\beta}_j),\Sigma)$, then
\begin{align*}
&\log h_j^\ast(\alpha \bk_1 + (1-\alpha)\bk_2,\alpha \Sigma_1+(1-\alpha) \Sigma_2) = \\
&\log h_j((Q (\alpha \bk_1+(1-\alpha)\bk_2)-\hat{\beta}_j),\alpha \Sigma_1 +(1-\alpha)\Sigma_2) = \\ & \log h_j((\alpha(Q\bk_1-\hat{\beta_j})+(1-\alpha)(Q\bk_2-\hat{\beta}_j),\alpha \Sigma_1 +(1-\alpha)\Sigma_2) \ge \\ &\alpha \log h_j(((Q\bk_1-\hat{\beta_j}),\Sigma_1)+(1-\alpha)\log h_j((Q\bk_2-\hat{\beta}_j),\Sigma_2)= \\ &\alpha \log h^\ast_j(\bk_1,\Sigma_1)+(1-\alpha)\log h^\ast_j(\bk_2,\Sigma_2)
\end{align*}
where the inequality follows from the logarithmic concavity of $h_j$ in $((Q \bk-\hat{\beta}_j),\Sigma)$. Thus, individual factors in the third term are logarithmically concave in $(\bk,\Sigma)$, and since logarithmic concavity is preserved under multiplication, the third term is also logarithmically concave in $(\bk,\Sigma)$. This combined with logarithmic concavity of $\pi_1$ from Proposition 2 shows that $\pi(\bk,\Sigma|D)$ is logarithmically concave in $(\bk,\Sigma)$. Since logarithmic concavity is preserved via marginalization, see \cite{pbekopa1971logarithmic} and \cite{leindler1972certain}, $\pi(\bk|D)=\int_{\mathbb{R}^{d(d+1)/2}}\pi(\bk,\Sigma|D) d \Sigma$ is logarithmically concave, and hence unimodal, in $\bk$.
\end{proof}
\section{Derivation of Metropolis-within-Gibbs sampler}
Here we provide some additional details on  the posterior calculations for Algorithm~1  in Section~3 by  deriving the required conditional distributions for $\kappa_k, \tau_k^2$, 
and $\Sigma$. 
\subsection*{Full conditionals for $\kappa_k$}
Defining $U:=nNQ^\top \Sigma^{-1} Q$ and $S:=nQ^\top  \Sigma^{-1} \sum_{j=1}^N\hat{\bbeta}_j$ the synthetic likelihood is 
\begin{align*}
SL(\bk,\Sigma|D)=&\prod_{j=1}^N (2 \pi)^{-d/2} |\Sigma/n|^{-1/2} \exp \{-1/2 (Q \bk-\hat{\beta}_j)^\top (\Sigma/n)^{-1} (Q \bk-\hat{\beta}_j) \} \\   \end{align*}
             
                 \begin{align*}   
                  =& (2 \pi/n)^{-Nd/2} [ \prod_{j=1}^N |\Sigma|^{-1/2} ] \exp\{-n/2 \sum_{j=1}^N \hat{\beta}_j^\top \Sigma^{-1} \hat{\beta}_j\}\exp \{-\bk^\top U \bk/2 + \bk^\top S \} \\
 =& (2 \pi/n)^{-Nd/2} [ \prod_{j=1}^N |\Sigma|^{-1/2} ] \exp\{-\sum_{m \ne k}\sum_{i \ne k}u_{im} \kappa_i \kappa_m+\sum_{i \ne k}s_i \kappa_i-n/2 \sum_{j=1}^N \hat{\beta}_j^\top \Sigma^{-1} \hat{\beta}_j\} \\ \times & \exp \{-u_{kk}\kappa_k^2/2 -\sum_{i \ne k}u_{ik}\kappa_i \kappa_k+ s_k \kappa_k \} \\
               =& C_k \exp \{-u_{kk}\kappa_k^2/2 +(s_k-\sum_{i \ne k}u_{ik}\kappa_i ) \kappa_k \}  
\end{align*}
where $C_k$ is the term depending on other parameters except $\kappa_k$. The first mixture (i.e., degenerate at 0) component's full conditional of $\kappa_k$ we have
\begin{equation*}
c_{1}=\int SL(\bk,\Sigma|D)\mathbb{I}_0(\kappa_k)d\kappa_k=C_k
\end{equation*}
For the second component
\begin{align}
c_2=&\int_{\mathbb{R}^+} SL(\bk|D)f(\kappa_k|\tau_k)d\kappa_k \nonumber \\
     =&C_k \int_{\mathbb{R}^+} \exp \{-u_{kk}\kappa_k^2/2 +(s_k-\sum_{i \ne k}u_{ik}\kappa_i ) \kappa_k \} \sqrt{\frac{2 }{\pi \tau_k^2}}\exp\{-  \kappa_k^2/2\tau_k^2 \}\mathbb{I}_{\mathbb{R}^+}(\kappa_k) d \kappa_k \nonumber \\
     =&C_k \sqrt{\frac{2 }{\pi \tau_k^2}} \int_{\mathbb{R}^+} \exp \{-(u_{kk}+1/\tau_k^2)\kappa_k^2/2 +(s_k-\sum_{i \ne k}u_{ik}\kappa_i ) \kappa_k \}\mathbb{I}_{\mathbb{R}^+}(\kappa_k) d \kappa_k \label{eq:FCtau} \\
     =&C_k \sqrt{\frac{2 }{\pi \tau_k^2}}\exp\{ \frac{(s_k-\sum_{i \ne k}u_{ik}\kappa_i )^2}{2(u_{kk}+1/\tau_k^2)} \} \int_{\mathbb{R}^+} \exp \{-(u_{kk}+1/\tau_k^2)(\kappa_k-\frac{(s_k-\sum_{i \ne k}u_{ik}\kappa_i)}{(u_{kk}+1/\tau_k^2)})^2/2 \} \mathbb{I}_{\mathbb{R}^+}(\kappa_k)d \kappa_k \label{eq:FCkappa} \\
     =&C_k \sqrt{\frac{2 }{\pi \tau_k^2}}\exp\{ \frac{(s_k-\sum_{i \ne k}u_{ik}\kappa_i )^2}{2(u_{kk}+1/\tau_k^2)} \} \sqrt{\frac{2 \pi}{(u_{kk}+1/\tau_k^2)}}(1-\Phi(0,\frac{(s_k-\sum_{i \ne k}u_{ik}\kappa_i)}{(u_{kk}+1/\tau_k^2)},(u_{kk}+1/\tau_k^2)^{-1})) \nonumber \\
     =&2 C_k \sqrt{\frac{1}{\tau_k^2(u_{kk}+1/\tau_k^2)}}\exp\{ \frac{(s_k-\sum_{i \ne k}u_{ik}\kappa_i)^2}{2(u_{kk}+1/\tau_k^2)} \} (1-\Phi(0,\frac{(s_k-\sum_{i \ne k}u_{ik}\kappa_i)}{(u_{kk}+1/\tau_k^2)},(u_{kk}+1/\tau_k^2)^{-1})) \nonumber
\end{align}
where $\Phi(x,a,b)$ is the normal cdf with mean $a$ and variance $b$ evaluated at $x$. So the full conditional for the weight is given as, 
\begin{align*}
\omega_k^\ast=&\frac{c_2 \omega_k}{c_1 (1-\omega_k)+c_2 \omega_k} \\
                       =& \frac{2  \sqrt{\frac{1}{\tau_k^2(u_{kk}+1/\tau_k^2)}}\exp\{ \frac{(s_k-\sum_{i \ne k}u_{ik}\kappa_i)^2}{2(u_{kk}+1/\tau_k^2)} \} (1-\Phi(0,\frac{(s_k-\sum_{i \ne k}u_{ik}\kappa_i )}{(u_{kk}+1/\tau_k^2)},(u_{kk}+1/\tau_k^2)^{-1})) \omega_k}{ (1-\omega_k)+ 2  \sqrt{\frac{1}{\tau_k^2(u_{kk}+1/\tau_k^2)}}\exp\{ \frac{(s_k-\sum_{i \ne k}u_{ik}\kappa_i )^2}{2(u_{kk}+1/\tau_k^2)} \} (1-\Phi(0,\frac{(s_k-\sum_{i \ne k}u_{ik}\kappa_i)}{(u_{kk}+1/\tau_k^2)},(u_{kk}+1/\tau_k^2)^{-1})) \omega_k} \\
                       =& \frac{\omega_k}{(1-\omega_k)/(2  \sqrt{\frac{1}{\tau_k^2(u_{kk}+1/\tau_k^2)}}\exp\{ \frac{(s_k-\sum_{i \ne k}u_{ik}\kappa_i)^2}{2(u_{kk}+1/\tau_k^2)} \} (1-\Phi(0,\frac{(s_k-\sum_{i \ne k}u_{ik}\kappa_i)}{(u_{kk}+1/\tau_k^2)},(u_{kk}+1/\tau_k^2)^{-1}))) + \omega_k}
\end{align*}
For the full conditional, $\pi_2(\kappa_k|\dots)$, the integrand in  \eqref{eq:FCkappa} implies that
\begin{equation}
\pi_2(\kappa_k|\dots) \sim \mathcal{N}(\frac{(s_k-\sum_{i \ne k}u_{ik}\kappa_i)}{(u_{kk}+1/\tau_k^2)},(u_{kk}+1/\tau_k^2)^{-1})\mathbb{I}_{\mathbb{R}^+}(\kappa_k)
\end{equation}
which is the normal distribution truncated at zero.
\subsection*{Full conditional for $\tau_k^2$}
From the expression in \eqref{eq:FCtau} it is clear that
\begin{align}
\pi(\tau_k^2|\dots) \propto& (\frac{1}{\tau_k^2})^{1/2}\exp \{\frac{-\kappa_k/2}{\tau_k^2} -\lambda_k^2 \tau_k^2/2 \} 
\end{align}
Define $s=\frac{1}{\tau_k^2}$, then $\tau_k^2=\frac{1}{s}$ and the Jacobian of the transformation is $|J|=\frac{1}{s^2}$ so that 
\begin{align*}
\pi(s|\dots) \propto & s^{-3/2}\exp \{\frac{-\kappa_ks^2-\lambda_k^2}{2s} \} \\
                  \propto & s^{-3/2}\exp \{\frac{-s^2+2\frac{\lambda_k}{\kappa_k}s-(\frac{\lambda_k}{\kappa_k})^2}{2\frac{1}{\kappa_k^2}s} \} \\
                             =&s^{-3/2}\exp \{\frac{-(s-\frac{\lambda_k}{\kappa_k})^2}{2\frac{1}{\kappa_k^2}s} \} \\
                            =&s^{-3/2}\exp \{\frac{-\lambda_k^2(s-\frac{\lambda_k}{\kappa_k})^2}{2\frac{\lambda_k^2}{\kappa_k^2}s} \}
\end{align*}
which is the inverse Gaussian distribution with mean $\frac{\lambda_k}{\kappa_k}$ and shape $\lambda_k^2$. To sample from the full conditional, $\tau_k^2|\dots$, one would first sample from the appropriate inverse Gaussian and then invert. Inverse Gaussian sampling is standard in statistical software packages such as~R. 
\subsection*{Full conditional for $\Sigma$}
It is clear that the target is proportional to $\pi(\bk,\Sigma|D)$, as defined Equation (12) of the manuscript. Then for any proposal $\Sigma^\ast \sim q(\Sigma^\ast|\Sigma)$, the Metropolis step accepts $\Sigma^\ast$ with probability $\min \bigg \{1, \frac{\pi(\bk,\Sigma^\ast|D)q(\Sigma|\Sigma^\ast)}{\pi(\bk,\Sigma|D)q(\Sigma^\ast|\Sigma)} \bigg \}$.
\section{R code implementation}
The Gibbs function below may be sourced into R and used to perform reaction network inference. The function takes as its arguments  (i) betas:  either those based on LSE or MEF,  (ii) $Q$: the matrix from the linear combination of kinetic parameters,  (iii) iter: the number of iterations of the sampler, and (iv) $n:$ the system size. The function output is the matrix containing posterior samples of the weights (omegas), kinetic parameters (kappas), and acceptance for the Metropolis step. Hyperparameters that can be modified: (v $\ge$ N+d+1)=degrees of freedom of the Wishart prior, weights[1,]=rep(0.5,r) may be modified to change prior edge probabilities. The proposal degrees of freedom, vprop, may be modified to vary the acceptance rate.    \\ \\
\#\# Requires betas, Q, iter, n=system size \\
Gibbs=function(betas,Q,iter,n){ \\
require(statmod); \\
require(truncnorm);\\
require(MCMCpack);\\
\\
accept=rep(0,iter); \\
r=length(Q[1,]);\\ 
d=length(Q[,1]);\\
N=2;\\
if(is.vector(betas))\{N=1; \\
d=length(betas); \\
sumBeta=rep(0,d); \\
sumBeta=betas;\} \\ \\
if(N$>$1)\{N=length(betas[,1]); \\
d=length(betas[1,]); \\
sumBeta=rep(0,d) \\
	for(j in 1:N)\{sumBeta=sumBeta+betas[j,];\} \\
\} \\
v=N+d+1; \\
vprop=n; \\
T=N*t(Q)\%*\%diag(1,d)\%*\%Q*n; \\
S=t(Q)\%*\%diag(1,d)\%*\%sumBeta*n; \\
weights=matrix(0,nrow=iter,ncol=r); \\
kappas=matrix(0,nrow=iter,ncol=r); \\
lambda=matrix(0,nrow=iter,ncol=r); \\
weights[1,]=rep(.5,r) \\
lambda=(1-weights[1,])/weights[1,]; \\
kappas[1,]=rep(1,r); \\
temp=c(); \\
tau=rep(1,r); \\
theta=10000; \\
Psi=diag(1/theta,d); \\
invPsi=diag(theta,d); \\
if(N$>$1){ \\
cv=(cov(betas)+Psi)/n; \\
} \\
else cv=(diag(var(betas),d)+Psi)/n; \\
\\
invCV=solve(cv)/vprop; \\
Sigma=diag(1,d); \\
detSigma=det(Sigma); \\
iSig=solve(Sigma); \\
for(i in 2:iter){ \\
	temp=kappas[(i-1),];	\\
	for(k in 1:r)\{ \\
		\\
		if(k==1)\{a=(S[k]-sum(T[(k+1):r,k]*kappas[(i-1),(k+1):r]))/(T[k,k]+1/tau[k]);\} \\
		if(k==r)\{a=(S[k]-sum(T[1:(r-1),k]*kappas[i,1:(r-1)]))/(T[k,k]+1/tau[k]);	\} \\	
				\\
	    if(k$>$1 \& k$<$r) \{a=(S[k]-sum(T[(k+1):r,k]*kappas[(i-1),(k+1):r])-sum(T[1:(k-1),k]*kappas[i,1:(k-1)]))/(T[k,k]+1/tau[k]);\} \\
		b=1/(T[k,k]+1/tau[k]); \\
		c2=1-pnorm(0,mean=a,sd=sqrt(b)); \\
		M=2*sqrt(b*(1/tau[k]))*exp(a\^{}2/(2*b))*c2; \\
                if(is.na(M))\{M=0;\}  \\
		weights[i,k]=weights[1,k]/((1-weights[1,k])/M+weights[1,k]); \\
		u=runif(1); \\
		if(u $<$ weights[i,k])\{ \\
                  kappas[i,k]=rtruncnorm(1,0,Inf,mean=a,sd=sqrt(b)); \\
                  tau[k]=1/rinvgauss(1,mean=lambda[k]/kappas[i,k],lambda[k]\^{}2);			\} \\
		else {kappas[i,k]=0;\} \\
		\\
	\} \\
	scale=diag(0,d) \\
	\\
if(N$>$1)\{ \\
for(j in 1:N)\{scale=scale+n*(Q\%*\%kappas[i,]-betas[j,])\%*\%t(Q\%*\%kappas[i,]-betas[j,]);\} \\
   \\
   \\
	Sigmaprop=rwish(vprop,Sigma/vprop); \\
	svdSigProp=svd(Sigmaprop); \\
	detSigmaprop=prod(svdSigProp\$d) \\
	iSigmaprop=svdSigProp\$v\%*\%diag(1/svdSigProp\$d)\%*\%t(svdSigProp\$u); \\
	\\
	u=runif(1); \\
    alpha=min(1,(detSigmaprop/detSigma)\^{}(-vprop+v/2-N/2)*exp(.5*vprop*sum(diag((iSig-invCV)\%*\%Sigmaprop-(iSigmaprop-invCV)\%*\%Sigma)))*exp(-.5*sum(diag(scale\%*\%(iSigmaprop-iSig))))); \\
    \\
    if(is.na(alpha))\{alpha=0;\} \\
    if(u$<$alpha)\{Sigma=Sigmaprop; iSig=iSigmaprop; detSigma=detSigmaprop; accept[i]=1;\} \\
\\
	T=N*t(Q)\%*\%iSig\%*\%Q*n; \\
	S=t(Q)\%*\%iSig\%*\%sumBeta*n; \\
}\\
if(N==1)\{ \\  
        scale=n*(Q\%*\%kappas[i,]-betas)\%*\%t(Q\%*\%kappas[i,]-betas);      \\
        Sigmaprop=rwish(vprop,Sigma/vprop); \\
	svdSigProp=svd(Sigmaprop); \\
	detSigmaprop=prod(svdSigProp\$d) \\
	iSigmaprop=svdSigProp\$v\%*\%diag(1/svdSigProp\$d)\%*\%t(svdSigProp\$u); \\
	\\
	u=runif(1); \\
    alpha=min(1,(detSigmaprop/detSigma)\^{}(-vprop+v/2-N/2)*exp(.5*vprop*sum(diag((iSig-invCV)\%*\%Sigmaprop-(iSigmaprop-invCV)\%*\%Sigma)))*exp(-.5*sum(diag(scale\%*\%(iSigmaprop-iSig))))); \\
    \\
    if(is.na(alpha))\{alpha=0;\} \\
    if(u$<$alpha)\{Sigma=Sigmaprop; iSig=iSigmaprop; detSigma=detSigmaprop; accept[i]=1;\} \\
\\
	T=N*t(Q)\%*\%iSig\%*\%Q*n; \\
	S=t(Q)\%*\%iSig\%*\%sumBeta*n; \\
  \}\\
  \\
    \}\\
Gibbs=cbind(weights,kappas,accept);\\
\}\\

\bibliographystyle{unsrt}
\bibliography{LFMCMC}

\end{document}